\documentclass{llncs}

\usepackage{amsmath}
\usepackage{algorithm}
\usepackage[noend]{algpseudocode}
\usepackage{multicol}

\begin{document}

\mainmatter

\title{A Fault-Tolerant Sequentially Consistent DSM\\With a Compositional Correctness Proof}

\titlerunning{A Fault-Tolerant Sequentially Consistent DSM}  

\author{Niklas Ekstr\"om \and Seif Haridi}

\authorrunning{Niklas Ekstr\"om and Seif Haridi} 

\tocauthor{Niklas Ekstr\"om and Seif Haridi}
\institute{KTH Royal Institute of Technology, Stockholm, Sweden\\
\email{\{neks,haridi\}@kth.se}}

\maketitle

\begin{abstract}
We present the SC-ABD algorithm that implements sequentially consistent distributed shared memory (DSM).
The algorithm tolerates that less than half of the processes are faulty (crash-stop).
Compared to the multi-writer ABD algorithm, SC-ABD requires one instead of two round-trips of communication to perform a write operation, and an equal number of round-trips (two) to perform a read operation.
Although sequential consistency is not a compositional consistency condition, the provided correctness proof is compositional.
\end{abstract}

\section{Introduction}

Using fault-tolerant distributed shared memory (DSM) as a building block in the design of a distributed system can simplify the design, as individual process failures are masked through replication.
To characterize an implementation of distributed shared memory, we consider the following criteria:
{
\renewcommand*{\thefootnote}{$\star$}
\footnotetext{This work was supported by the Swedish Foundation for Strategic Research (SSF).}
}
\begin{itemize}
\item Consistency: a stronger consistency condition may be easier to program against, but may provide worse performance, and vice versa.
\item Multiple writers: an implementation may allow a single process, or multiple processes, to update registers.
\item Latency: the number of round-trips of communication required to execute an operation.
\item Resilience: the number of processes that can be tolerated to be faulty in an execution, $f$, in relation to the total number of processes in the system, $n$.
\end{itemize}


In this paper, we consider the problem of implementing distributed shared memory that is sequentially consistent, allow multiple writers, can complete a write operation after one round of communication and a read operation after two rounds of communication, and that tolerates $f<n/2$ faulty processes.
We present the SC-ABD algorithm as a solution to this problem.
In Table~\ref{comparison} in the conclusion section, we present a comparison of SC-ABD to two other DSM algorithms along the mentioned criteria.

Proving that a distributed shared memory implementation satisfies sequential consistency can be a difficult task.
Unlike some other consistency conditions, sequential consistency is not a \emph{compositional} consistency condition.
Never the less, the proof given for the correctness of SC-ABD is compositional, and we therefore present this proof technique as a contribution in itself.

\section{Model and Definitions}
We consider an asynchronous distributed system composed of $n$ processes, denoted $p_1, \dots, p_n$, and a communication network with reliable links.
We denote by $\Pi = \{1, \dots, n\}$ the set of process identifiers.
In any given system execution, a process is said to be correct if the process never crashes, and otherwise it is said to be faulty.
A process that crashes stops taking steps and can never recover.
We assume that at most $f$ processes are faulty in any given execution, where $f < n/2$.

\subsection{Shared Memory}

A distributed shared memory is a distributed implementation of shared memory.
We consider a shared memory consisting of read/write registers.
Each register holds an integer value, initially zero.
The shared memory defines a set of primitive operations, that provide the only means to manipulate the registers.
In our case, the operations provided are \emph{read} and \emph{write}.
A process invokes an operation and receives a response when the execution of the operation is complete.
We will refer to an \emph{operation execution} as an operation, if the distinction is clear from the context.
Each process is allowed to have at most one outstanding operation, meaning that a process may not invoke another operation before the process has received the response for the previously invoked operation.
Let $o$ refer to a particular operation execution, invoked by process $p_i$.
We denote by $\mathit{inv}(o)$ the \emph{invocation event} that occurs when $p_i$ invokes $o$, and denote by $\mathit{res}(o)$ the \emph{response event} that occurs when the execution of $o$ completes.

We model an execution using a \emph{history}, which is a sequence of invocation and response events, ordered by the real times when the events occurred.
History $H$ is \emph{sequential} if the first event is an invocation event, and every invocation event (except possibly the last) is immediately followed by the matching response event.
By $H|p_i$ we denote the subsequence of $H$ where every event occurs in process $p_i$; we refer to $H|p_i$ as a \emph{process subhistory}.
Similarly, by $H|x$ we denote the subsequence of $H$ containing only events related to operations that target register $x$, and refer to $H|x$ as a \emph{register subhistory}.
A history is \emph{well-formed} if each process subhistory is a sequential history, and in the following we only consider well-formed histories.
Two histories $H$ and $H'$ are \emph{equivalent}, denoted $H \simeq H'$, if and only if, for each process $p_i$, $H|p_i = H'|p_i$.
For events $e_1$ and $e_2$ in history $H$ we write $e_1 <_H e_2$ to denote that $e_1$ precedes $e_2$ in $H$.
We say that ``operation $o$ is in history $H$'' if $\mathit{inv}(o)$ is in $H$.
For operations $o_1$ and $o_2$ in $H$ we write $o_1 <_H o_2$ to denote that $res(o_1) <_H inv(o_2)$.

Operation $o$ is \emph{pending} in history $H$ if the invocation event for $o$ is in $H$ but not the response event.
History $H$ is complete if $H$ does not contain any pending operations.
For presentational simplicity, we consider only complete histories in the rest of this paper.

The shared memory has a \emph{sequential specification}, which is a set containing all sequential histories such that each read operation of some register returns the value written by the last write to that register (the write closest preceding the read in the sequential history), or the default value if no such write exists.
A sequential history is \emph{legal} if it is in the shared memory's sequential specification.

Sequential consistency is a consistency condition that was described by Lamport~\cite{Lamport:1979:MMC:1311099.1311750}.
We define what it means for a history to be sequentially consistent:

\begin{definition} \label{sequentially_consistent}
History $H$ is \emph{sequentially consistent}, denoted $\textsf{SC}(H)$, if and only if there exists a legal sequential history $S$ such that $S \simeq H$.
\end{definition}

The correctness conditions that we require of an algorithm implementing sequentially consistent distributed shared memory are:
\begin{itemize}
\item \textbf{Termination:} If a correct process invokes an operation, then the operation eventually completes.
\item \textbf{Sequential Consistency:} Each history corresponding to an execution of the algorithm must be sequentially consistent.
\end{itemize}

\subsection{Causality and Logical Clocks}

Causality and logical clocks were described in a paper by Lamport~\cite{Lamport:1978:TCO:359545.359563}.
Event $e_1$ is said to \emph{causally precede} event $e_2$, denoted $e_1 \rightarrow e_2$, if at least one of the following conditions hold:
(1) $e_1$ and $e_2$ both occur in the same process and $e_1$ occurs before $e_2$, (2) $e_1$ is the sending of message $m$ and $e_2$ is the receipt of $m$, (3) there exists an event $e'$ such that $e_1 \rightarrow e'$ and $e' \rightarrow e_2$.

A logical clock is a device that assigns integers to events in a manner consistent with the causally precedes relation.
More precisely, by letting $lt(e)$ denote the logical time assigned to event $e$, we require that: $e_1 \rightarrow e_2 \Rightarrow  lt(e_1) < lt(e_2)$.

\section{Algorithm}

In this section we present the SC-ABD algorithm, whose pseudo-code is contained in Algorithm~\ref{scabd}.
The algorithm is given as a set of reactive \emph{handlers}.
Each handler has an associated condition that describes when that handler is eligible for execution, e.g., when an operation is invoked, or a message is received.

For each process, the algorithm contains a variable $lt$ that implements a logical clock.
Whenever a handler is executed in response to a local condition (i.e., an operation is invoked) the logical clock is incremented by one.
When a message is sent from process $p_i$ to process $p_j$, the current logical time of $p_i$ is included in the message, and when the message is received by $p_j$ and the corresponding handler is executed, $p_j$'s logical clock is updated to a logical time that is one greater than the maximum of $p_j$'s previous logical time and the logical time included in the message.

Each process stores the values that have been written to the registers.
In order to determine which value is more recent, a timestamp is associated with each value.
A value and its associated timestamp are stored together as a \emph{timestamp-value pair}.
The algorithm has a local variable, $\mathit{tvps}$, that maps register identifiers to timestamp-value pairs.

Communication in the algorithm proceeds in \emph{phases}.
A phase consists of a round of communication, where the process executing the phase, $p_i$, sends a request to all processes and waits for responses from a majority of the processes before the phase ends.

A write operation has one phase: the update phase.
The process executing the write operation, $p_i$, creates a timestamp as the pair with $p_i$'s current logical time and $p_i$'s process identifier, $i$.
It then pairs this timestamp together with the value to be written into a timestamp-value pair.
$p_i$ sends an update request containing the register identifier and the timestamp-value pair to all processes (lines 16-20 in Algorithm~\ref{scabd}).
When process $p_j$ receives the update request it updates its $\mathit{tvps}$ with the supplied timestamp-value pair if the timestamp is greater than the timestamp of the timestamp-value pair that was previously stored, and then sends an ack response (lines 21-23).
After $p_i$ receives acks from a majority of processes, $p_i$ returns OK (lines 24-30).

A read operation has two phases: the query phase and the update phase.
The process executing the read operation, $p_i$, sends a query request to all processes containing the register identifier  for the register that is being read (lines 1-5).
When process $p_j$ receives the query request, $p_j$ retrieves the timestamp-value pair stored in $tvps$ for the register identifier, and sends this timestamp-value pair in a response message to $p_i$.
This timestamp-value pair is the maximal timestamp-value pair that $p_j$ has received so far in an update request, or the initial timestamp-value pair, $((0,0), 0)$, if no update request had been received previously (lines 6-7).
When $p_i$ has received response messages from a majority of processes, $p_i$ chooses the timestamp-value pair, $(ts,v)$, with the maximum timestamp out of the timestamp-value pairs received.
Before returning value $v$, $p_i$ performs an update phase using the $(ts, v)$ timestamp-value pair, in order to guarantee that a majority of the processes have stored the timestamp-value pair before the read completes (lines 8-15 and 21-30).

\begin{algorithm}
	\caption{\textsc{SC-ABD} -- code for $p_i$.}\label{scabd}
	
	\begin{algorithmic}[1]
		\Statex \textbf{Local variables:}
		\Statex $\mathit{lt}$ -- logical time; initially 0
		\Statex $\mathit{rid}$ -- current request identifier; initially 0
		\Statex $\mathit{tvps}$ -- map from register ids to timestamp-value pairs; initially maps to $((0, 0), 0)$
		\Statex $\mathit{responses}$ -- tracking responses/acks; initially $\{\}$
		\Statex $\mathit{reading}$ -- indicating whether currently reading ($\mathit{true}$) or writing ($\mathit{false}$)
		\Statex $\mathit{rreg}, \mathit{rval}$ -- temporary storage for register identifier and return value during reads
		\Statex
		\Statex \textbf{Note:} bcast $\langle m\rangle$ is an abbreviation for: \textbf{for} $j \in \Pi$ \textbf{do} send $\langle m\rangle$ to $p_j$
	\end{algorithmic}
	
	\begin{multicols}{2}
		\begin{algorithmic}[1]
			\Statex \textbf{When READ($\mathit{r}$) is invoked:}
			\State $\mathit{lt} \gets \mathit{lt} + 1$
			\State $\mathit{reading} \gets \mathit{true}$
			\State $\mathit{rreg} \gets \mathit{r}$
			\State $\mathit{rid} \gets rid + 1$
			\State bcast $\langle\mbox{``query''}, \mathit{lt}, \mathit{rid}, \mathit{r}\rangle$
			\Statex
			\Statex \textbf{When $\langle\mbox{``query''}, \mathit{lt}^\prime, \mathit{rid}^\prime, \mathit{r}\rangle$ is}
			\Statex \hspace{3mm}\textbf{received from $p_j$:}
			\State $\mathit{lt} \gets \max(\mathit{lt}, \mathit{lt}^\prime) + 1$
			\State send $\langle\mbox{``response''}, \mathit{lt}, \mathit{rid}^\prime, \mathit{tvps}[r]\rangle$ to $p_j$
			\Statex
			\Statex \textbf{When $\langle\mbox{``response''}, \mathit{lt}^\prime, \mathit{rid}^\prime, \mathit{tsv}^\prime\rangle$ is}
			\Statex \hspace{3mm}\textbf{received from $p_j$ with $\mathit{rid}=\mathit{rid}^\prime$:}
			\State $\mathit{lt} \gets \max(\mathit{lt}, \mathit{lt}^\prime) + 1$
			\State $\mathit{responses} \gets \mathit{responses} \cup \{ (\mathit{tsv}^\prime, j)\}$
			\If{$|\mathit{responses}| = \lfloor |\Pi| / 2 \rfloor + 1$}
			\State $(\mathit{tsv}, \_) \gets \max(\mathit{responses})$
			\State $(\mathit{ts}, \mathit{rval}) \gets \mathit{tsv}$
			\State $\mathit{responses} \gets \{\}$
			\State $\mathit{rid} \gets rid + 1$
			\State bcast $\langle\mbox{``update''}, \mathit{lt}, \mathit{rid}, \mathit{rreg}, \mathit{tsv}\rangle$
			\EndIf
			\Statex \textbf{When WRITE($\mathit{r}, \mathit{v}$) is invoked:}
			\State $\mathit{lt} \gets \mathit{lt} + 1$
			\State $\mathit{reading} \gets \mathit{false}$
			\State $\mathit{tsv} \gets ((lt, i), v)$
			\State $\mathit{rid} \gets rid + 1$
			\State bcast $\langle\mbox{``update''}, \mathit{lt}, \mathit{rid}, \mathit{r}, \mathit{tsv}\rangle$
			\Statex
			\Statex \textbf{When $\langle\mbox{``update''}, \mathit{lt}^\prime, \mathit{rid}^\prime, \mathit{r}, \mathit{tsv}^\prime\rangle$ is}
			\Statex \hspace{3mm}\textbf{received from $p_j$:}
			\State $\mathit{lt} \gets \max(\mathit{lt}, \mathit{lt}^\prime) + 1$
			\State $\mathit{tvps}[r] \gets \max(\mathit{tvps}[r], \mathit{tsv}^\prime)$
			\State send $\langle\mbox{``ack''}, \mathit{lt}, \mathit{rid}^\prime\rangle$ to $p_j$
			\Statex
			\Statex \textbf{When $\langle\mbox{``ack''}, \mathit{lt}^\prime, \mathit{rid}^\prime\rangle$ is}
			\Statex \hspace{3mm}\textbf{received from $p_j$ with $\mathit{rid}=\mathit{rid}^\prime$:}
			\State $\mathit{lt} \gets \max(\mathit{lt}, \mathit{lt}^\prime) + 1$
			\State $\mathit{responses} \gets \mathit{responses} \cup \{j\}$
			\If{$|\mathit{responses}| = \lfloor |\Pi| / 2 \rfloor + 1$}
			\State $\mathit{responses} \gets \{\}$
			\State $\mathit{rid} \gets rid + 1$
			\State \textbf{if} $\mathit{reading}$ \textbf{then} RETURN $\mathit{rval}$
			\State \textbf{else} RETURN OK
			\EndIf
		\end{algorithmic}
	\end{multicols}
\end{algorithm}

\section{Correctness Proof} \label{sec_correctness}

We first prove that SC-ABD satisfies the termination property.

\begin{lemma} \label{lemma_termination}
Algorithm SC-ABD satisfies the termination property.
\end{lemma}

\begin{proof}
As links are reliable and a majority of processes are correct according to the assumptions in our model, each communication phase executed by a correct process is guaranteed to eventually complete, and every operation executed by a correct process is therefore guaranteed to complete.\qed
\end{proof}

In the rest of this section we prove that the algorithm satisfies sequential consistency.

\subsection{Linearizability}

Linearizability is a consistency condition described by Herlihy and Wing~\cite{Herlihy:1990:LCC:78969.78972}.

\begin{definition} \label{linearizable}
History $H$ is \emph{linearizable}, denoted $\textsf{LIN}(H)$, iff there exists a legal sequential history $S$ such that $S \simeq H$, and $\forall o_1,o_2 \in H: o_1 <_H o_2 \Rightarrow o_1 <_S o_2$.
\end{definition}

Linearizability is \emph{compositional}, in the sense that history $H$ is linearizable if and only if each register subhistory $H|x$ is linearizable:
\begin{equation} \label{lin_compositionality}
\textsf{LIN}(H) \Leftrightarrow \forall x: \textsf{LIN}(H|x)
\end{equation}

From the definition of sequential consistency and the definition of linearizability, it follows that linearizability is stronger than sequential consistency:

\begin{equation} \label{lin_stronger_sc}
\textsf{LIN}(H) \Rightarrow \textsf{SC}(H)
\end{equation}

\subsection{Logical-Time History}

We define the \emph{logical-time history} corresponding to history $H$, denoted $H^{lt}$, to be the sequence containing the same events as $H$, but reordered according to the logical times when the events occurred, using the process identifiers of the processes where the events occurred to break ties.
%


For each process $p_i$, the relative ordering of events in $H|p_i$ is preserved in $H^{lt}|p_i$, as the logical times of events in $H|p_i$ are monotonically increasing.
It follows that the (real-time) history $H$ and its corresponding logical-time history $H^{lt}$ are equivalent, $H \simeq H^{lt}$.
Together with the definition of sequentially consistent histories it follows that:
\begin{equation} \label{sc_equivalence}
\textsf{SC}(H) \Leftrightarrow \textsf{SC}(H^{lt})
\end{equation}

\subsection{Compositional Reasoning}

Combining (\ref{lin_compositionality}), (\ref{lin_stronger_sc}), and (\ref{sc_equivalence}), we have:
\begin{equation} \label{eq_ltlin_compositionality}
\left( \forall x:\textsf{LIN}(H^{lt}|x) \right) \Rightarrow \textsf{LIN}(H^{lt}) \Rightarrow \textsf{SC}(H^{lt}) \Rightarrow \textsf{SC}(H)
\end{equation}
Equation (\ref{eq_ltlin_compositionality}) allows us to reason compositionally, i.e., to reason about, for each register $x$, the register subhistory $H^{lt}|x$ in isolation.

\subsection{Reasoning about the Algorithm}

We state a couple of definitions regarding the algorithm:

\begin{itemize}
\item The logical time of a handler execution is the value assigned to the $lt$ variable on the handler's first line in the algorithm text.

\item The timestamp of operation $o$, denoted $ts(o)$, is the timestamp used in the operation's update phase.
\end{itemize}
\vspace{5mm}
From the definition of logical-time history $H^{lt}$, it follows that:
\begin{equation} \label{lt_assump}
o_1 <_{H^{lt}} o_2 \Rightarrow lt(res(o_1)) \leq lt(inv(o_2))
\end{equation}

We state and prove the following proposition:
\begin{proposition} \label{ts_leq}
Let $o_1$ and $o_2$ be operations in $H^{lt}|x$ such that $o_1$ contains an update phase and $o_2$ contains a query phase.
If $o_1 <_{H^{lt}|x} o_2$ then $ts(o_1) \leq ts(o_2)$. 
\end{proposition}

\begin{proof}
Let $p_i$ be the process that executes the update phase in $o_1$, and $p_j$ be the process that executes the query phase in $o_2$.
At the time when $p_i$'s update phase completes, $p_i$ will have received response messages from a majority of processes.
Let $M_u$ refer to this majority set of processes.
Similarly, let $M_q$ refer to the majority set of processes from which $p_j$ received responses before the query phase in operation $o_2$ completed.
As any two majority sets intersect, there must be one process, $p_k$, that is both in $M_u$ and in $M_q$.

Let $e_1$ be the event when $p_k$ processes $o_1$'s update request, and $e_2$ the event when $p_k$ processes $o_2$'s query request.
By causality we have $lt(e_1) < lt(res(o_1))$ and $lt(inv(o_2)) < lt(e_2)$, and together with (\ref{lt_assump}) we get $lt(e_1) < lt(e_2)$.
Since $e_1$ and $e_2$ are in the same process, this implies that $e_1$ occurs before $e_2$.

Since $p_k$ returns the timestamp-value pair with the maximal timestamp that it has received in all previous update requests, the timestamp in the response to $o_2$'s query request is guaranteed to be greater than or equal to the timestamp in $o_1$'s update request.
As $p_j$ picks the timestamp-value pair with the maximal timestamp on line 11 of the algorithm, and uses it in its update phase, it follows that $ts(o_1) \leq ts(o_2)$. \qed
\end{proof}

\begin{lemma} \label{lemma_sc}
Algorithm SC-ABD satisfies the sequential consistency property.
\end{lemma}

\begin{proof}
By using equation (\ref{eq_ltlin_compositionality}), we prove that the algorithm satisfies sequential consistency, by showing, for each execution, and for each register $x$, that $\textsf{LIN}(H^{lt}|x)$ holds.
From the definition of linearizability, we see that in order to prove that $\textsf{LIN}(H^{lt}|x)$ holds we are required to show that there exists a legal sequential history $S$ such that $S \simeq H^{lt}|x$, and, for all operations $o_1$ and $o_2$ in $H^{lt}|x$, if $o_1$ precedes $o_2$ in $H^{lt}|x$ then $o_1$ also precedes $o_2$ in $S$.
We proceed by creating a total order on the operations in $H^{lt}|x$ as follows:
\begin{enumerate}
\item Order write operations according to their timestamps. Any two write operations have unique timestamps by construction, so this is a total order.
\item Then order each read operation immediately after the write operation that wrote the value that the read operation returned.
If there are more than one read operations with the same timestamp then they are internally ordered based on the logical times when they were invoked (breaking ties using process identifiers).
\end{enumerate}
Let $S$ be the sequential history obtained from this total order.
As each read operation in $S$ returns the value written by the closest preceding write operation, it follows that $S$ is legal.

\vspace{10mm}
We show that $o_1 <_{H^{lt}|x} o_2 \Rightarrow o_1 <_S o_2$ using the following case analysis:

\begin{itemize}
\item $o_1$ is a write, $o_2$ is a write:
By causality we have $lt(inv(o_1)) < lt(res(o_1))$, which together with (\ref{lt_assump}) gives us $lt(inv(o_1)) < lt(inv(o_2))$.
Because of how the algorithm constructs timestamps (line 18), this implies that $ts(o_1) < ts(o_2)$, from which $o_1 <_S o_2$ follows.

\item $o_1$ is a read, $o_2$ is a write:
There exists a write $w_0$ such that $ts(w_0) = ts(o_1)$.
Since the invocation event of $w_0$ causally precedes the response event of $o_1$, we have $lt(inv(w_0)) < lt(res(o_1))$, and, using (\ref{lt_assump}), we have $lt(inv(w_0)) < lt(inv(o_2))$.
From the analysis of the previous case we have $ts(o_1) = ts(w_0) < ts(o_2)$, from which $o_1 <_S o_2$ follows.

\item $o_1$ is a write, $o_2$ is a read:
By the assumption and Proposition~\ref{ts_leq} it follows that $ts(o_1) \leq ts(o_2)$, from which $o_1 <_S o_2$ immediately follows.

\item $o_1$ is a read, $o_2$ is a read:
Again, by the assumption and Proposition~\ref{ts_leq} it follows that $ts(o_1) \leq ts(o_2)$.
If $ts(o_1) < ts(o_2)$ we directly have $o_1 <_S o_2$.
Otherwise, we have $ts(o_1) = ts(o_2)$. By causality and (\ref{lt_assump}) we have $lt(inv(o_1)) < lt(inv(o_2))$, and $o_1 <_S o_2$ follows from the definition of $S$.
\end{itemize}

Finally we must show that $S \simeq H^{lt}|x$.
For any process $p_i$, consider the history $(H^{lt}|x)|p_i$, which is sequential.
For any pair of operations $o_1$ and $o_2$ in $(H^{lt}|x)|p_i$, either $o_1 <_{(H^{lt}|x)|p_i} o_2$ or $o_2 <_{(H^{lt}|x)|p_i} o_1$.
The same ordering will be preserved in $S|p_i$, according to the case analysis above.
As $S$ and $H^{lt}|x$ contain the same events, we have $S \simeq H^{lt}|x$.\qed
\end{proof}

\begin{theorem}
Algorithm SC-ABD is a correct implementation of sequentially consistent distributed shared memory.
\end{theorem}

\begin{proof}
Follows directly from Lemma~\ref{lemma_termination} and Lemma~\ref{lemma_sc}.\qed
\end{proof}

\section{Related Work}

Research about shared memory has a long history in distributed computing.

\subsection{Consistency Conditions}

Lamport described sequential consistency~\cite{Lamport:1978:TCO:359545.359563}.
In multiprocessor systems, sequential consistency is widely regarded as the ``gold standard'',
but most multiprocessor systems provide weaker consistency by default, and require that programs use memory fences to achieve sequentially consistent behavior.

Proving that a shared memory implementation satisfies sequential consistency is a well-researched problem.
Alur, McMillan, and Peled proved that, in general, the sequential consistency verification problem is undecidable~\cite{Alur:1996:MCC:788018.788833}.

Bingham, Condon, and Hu suggested that the original formulation of sequential consistency, which is not prefix-closed, may be a reason why the verification problem is hard, and suggested two alternative variants to sequential consistency, Decisive Sequential Consistency (DSC) and Past-Time Sequential Consistency (PTSC) that are prefix-closed~\cite{Bingham:2003:TDN:777412.777467}.

Plakal, Sorin, Condon, and Hill use logical (Lamport) clocks as a tool to reason about correctness of their distributed shared memory protocol~\cite{Plakal:1998:LCV:277651.277672}.

Linearizability was described by Herlihy and Wing~\cite{Herlihy:1990:LCC:78969.78972}.
Linearizability has the pleasant property that it is a compositional consistency condition.

The cost of sequential consistency vs. linearizability was analyzed by Attiya and Welch~\cite{Attiya:1994:SCV:176575.176576}.
They proved that the cost of sequential consistency is lower than the cost of linearizability under reasonable assumptions.

\subsection{Fault-Tolerant Shared Memory}

The ABD algorithm was described by Attiya, Bar-Noy and Dolev~\cite{Attiya:1995:SMR:200836.200869}.
ABD was the first algorithm that showed it to be possible to implement fault-tolerant linearizable shared memory in a message passing system, but allowed only a single process to write to the memory.
Write operations complete after a single round of communication and read operations complete after two rounds.

The multi-writer ABD (MW-ABD) algorithm was described by Lynch and Shvartsman~\cite{Lynch:1997:RES:795670.796859}.
MW-ABD extended the ABD algorithm by allowing multiple processes to write to the memory, and in order to do so added a second round of communication to write operations.

\section{Conclusion}

We presented the SC-ABD algorithm that implements fault-tolerant, sequentially consistent, distributed shared memory, and proved it to be correct using a compositional proof structure.

Table~\ref{comparison} contains a comparison between SC-ABD, ABD, and MW-ABD along the criteria mentioned in the introduction:
consistency condition (linearizability (LIN) or sequential consistency (SC)); multiple writers allowed; number of rounds of communication required to complete a write (W)/read (R) operation; and how many faulty processes, $f$, that the algorithm tolerates.

\begin{table}[]
	\centering
	\caption{Comparison between three fault-tolerant DSM algorithms.}
	\label{comparison}
	\begin{tabular}{l|l|l|l}
		& ABD      & MW-ABD   & SC-ABD   \\ \hline
		Consistency      & LIN      & LIN      & SC       \\ \hline
		Multiple writers & No       & Yes      & Yes      \\ \hline
		Latency          & W:1, R:2 & W:2, R:2 & W:1, R:2 \\ \hline
		Resilience       & $f<n/2$  & $f<n/2$  & $f<n/2$  \\ \hline
	\end{tabular}
\end{table}

In a situation where an application, running on top of distributed shared memory, would satisfy its correctness conditions if the distributed shared memory provides sequential consistency, and the application would benefit from having a lower latency for write operations, we think that SC-ABD is a good choice.

Finally, we showed that, although sequential consistency is not a compositional consistency condition, it was still possible to reason compositionally about the correctness of the algorithm.

\section*{Acknowledgements}

We would like to thank the Swedish Foundation for Strategic Research for funding this work, and Jingna Zeng for helpful discussions.

\bibliographystyle{abbrv}
\bibliography{article_refs}

\end{document}